\documentclass[conference,letterpaper]{IEEEtran}

\usepackage[dvipsnames]{xcolor}

\usepackage{amsthm}
\usepackage[utf8]{inputenc} 
\usepackage[T1]{fontenc}
\usepackage{url}
\usepackage{ifthen}
\usepackage{cite}
\usepackage[cmex10]{amsmath} 

\usepackage{algorithm}
\usepackage[noend]{algpseudocode}

\usepackage{mathtools}
\usepackage{amssymb}
\usepackage{xspace}
\usepackage[nameinlink,capitalize]{cleveref}
\usepackage{balance}

\newtheorem{theorem}{Theorem}

\newtheorem{lemma}[theorem]{Lemma}

\newtheorem{remark}[theorem]{Remark}
\newtheorem{example}[theorem]{Example}

\newtheorem{proposition}[theorem]{Proposition}

\newtheorem{construction}[theorem]{Construction}

\Crefname{equation}{Eq.}{Eqs.}
\Crefname{figure}{Fig.}{Figs.}
\Crefname{construction}{Constr.}{Constrs.}
\Crefname{example}{Ex.}{Exs.}

\usepackage{mleftright}

\usepackage{array}
\usepackage{siunitx,booktabs}
\usepackage{multirow}

\usepackage{tikz, pgfplots}
\pgfplotsset{compat=1.18}

\usepackage{subcaption}
\usepackage{bm}

\newcommand{\6}{\mathbf}

\newcommand{\code}{\mathcal{C}}
\newcommand{\C}{\mathcal{C}}
\newcommand{\F}{\mathbb{F}}
\newcommand{\AmbSpace}{\mathcal{F}}
\newcommand{\NN}{\mathbb{N}}

\newcommand{\bLw}{{\mathcal{B}_{\bLam}}}

\newcommand{\e}{\6e}

\newcommand{\cstar}{\6c^\star}
\newcommand{\jstar}{j^\star}

\newcommand{\codeA}{\mathcal{A}}

\newcommand{\codeB}{\mathcal{B}}

\newcommand{\bn}{\6n}
\newcommand{\bLam}{\bm{\lambda}}

\newcommand{\dirprod}{\times}

\newcommand{\enc}[1]{\mathsf{enc}(#1)}
\newcommand{\encinv}[1]{\mathsf{enc}^{-1}(#1)}
\newcommand{\dec}[1]{\mathsf{dec}(#1)}
\newcommand{\GMDdec}[1]{\mathsf{GMD\text{-}dec}\mleft(#1\mright)}
\newcommand{\nC}[1]{n(#1)}
\newcommand{\kC}[1]{k(#1)}
\newcommand{\dHC}[1]{d_\mathrm{H}(#1)}
\newcommand{\dLC}[1]{d_{\bLam}(#1)}

\newcommand{\tLC}[1]{t_{\bLam}(#1)}
\newcommand{\LBtLC}[1]{{t}'_{\bLam}(#1)}
\newcommand{\LBdLC}[1]{{d}'_{\bLam}(#1)}

\newcommand{\size}[1]{\left\lvert#1\right\rvert}

\newcommand{\wH}[1]{w_{\scriptscriptstyle \mathrm{H}}(#1)}
\newcommand{\dH}[2]{d_{\scriptscriptstyle \mathrm{H}}(#1,#2)}

\newcommand{\wN}[1]{w_{\bn}(#1)}
\newcommand{\wL}[1]{w_{\bLam}(#1)}
\newcommand{\dL}[2]{d_{\bLam}(#1,#2)}
\newcommand{\tL}[1]{t_{\bLam}(#1)}

\newcommand{\ra}{{\6r_{\6a}}}
\newcommand{\rb}{{\6r_{\6b}}}

\newcommand{\Uset}{\mathcal{U}}
\newcommand{\Vset}{\mathcal{V}}
\newcommand{\Eset}{\mathcal{E}}
\newcommand{\Gset}{\mathcal{G}}
\newcommand{\Gsetstar}{\mathcal{G}^\star}

\newcommand{\extdeg}{m}

\pgfplotsset{
  GVstyle/.style={
    line width=1pt,
    dashed,
    color=gray,
  },
  naiveGVstyle/.style={
    line width=1pt,
    densely dotted,
    color=gray,
  },
  LPstyle/.style={
    line width=1pt,
    dashed,
    color=YellowOrange,
  },
  naiveLPstyle/.style={
    line width=1pt,
    densely dotted,
    color=YellowOrange,
  },
  Singletonstyle/.style={
    line width=1pt,
    dashed,
    color=CornflowerBlue,
  },
  naiveSingletonstyle/.style={
    line width=1pt,
    densely dotted,
    color=CornflowerBlue,
  },
  SPstyle/.style={
    line width=1pt,
    dashed,
    color=Green,
  },
  naiveSPstyle/.style={
    line width=1pt,
    densely dotted,
    color=Green,
  },  
  newCstyle/.style={
    mark = triangle*,
    color = RoyalBlue,
    only marks,
  },  
  oldCstyle/.style={
    mark = *,
    color = Red,
    only marks,
  }, 
}

\begin{document}
\title{Weighted-Hamming Metric: Bounds and Codes} 

\author{%
  \IEEEauthorblockN{\textbf{Sebastian Bitzer}$^1$, \textbf{Alberto Ravagnani}$^2$, and \textbf{Violetta Weger}$^1$}
  \IEEEauthorblockA{$^1$Technical University of Munich, Germany, $^2$Eindhoven University of Technology, the Netherlands \\
                    \textit{\{sebastian.bitzer, violetta.weger\}@tum.de}, 
                    \textit{a.ravagnani@tue.nl}
                    }
}

\maketitle

\begin{abstract}
The weighted-Hamming metric generalizes the Hamming metric by assigning different weights to blocks of coordinates.
It is well-suited for applications such as coding over independent parallel channels,
each of which has a different level of importance or noise.
From a coding-theoretic perspective, the actual error-correction capability of a code under this metric can exceed half its minimum distance.
In this work, we establish direct bounds on this capability, tightening those obtained via minimum-distance arguments.
We also propose a flexible code construction based on generalized concatenation and show that these codes can be efficiently decoded up to a lower bound on the error-correction capability.
\end{abstract}

\section{Introduction}

\textbf{Weighted-Hamming metric codes.} \ 
The weighted-Hamming metric (wHm) generalizes the ubiquitous Hamming metric by scaling the Hamming weight of $m$ blocks of coordinates of a vector by fixed weights $\lambda_1,\ldots,\lambda_m\in\NN$.
This assignment of weights can be used to reflect different importance levels of symbols or their vulnerability to errors, which makes the weighted-Hamming metric suitable to model a wide range of applications, including parallel independent channels \cite{bitzer2024weighted} and similarity search \cite{zhang2013binary}.
Nevertheless, the theory of wHm codes is still in its early stages and remains largely underdeveloped.

Previous works~\cite{moon2018weighted,bezzateev2013class} investigated the wHm properties of classical codes, and in~\cite{bitzer2024weighted} a systematic study was initiated.
Importantly, \cite{bitzer2024weighted} observes that the minimum weighted distance $\dLC{\C}$ does \textit{not} fully characterize the wHm error-correction capability $\tLC{\C}$ of the code $\C$, i.e., the number of uniquely correctable errors.
As a natural starting point, \cite{bitzer2024weighted} derives bounds on $\tLC{\C}$ through bounds on $\dLC{\C}$ and introduces a simple code construction.

\textbf{This work.} \ 
In this paper, we shift the focus 
from the minimum distance to the actual error-correction capability.
We directly provide bounds on $\tLC{\C}$ and, by avoiding the detour via the minimum distance, obtain tighter estimates for the number of correctable errors.
In addition, we introduce a flexible method for constructing wHm codes based on the framework of generalized code concatenation, which encompasses the coding scheme proposed in \cite{bitzer2024weighted}.
As a result, we obtain lower bounds on $\tLC{\C}$, denoted by $\LBtLC{\C}$, and even new bounds on $\dLC{\C}$, denoted by $\LBdLC{\C}$. 
These bounds indicate that, for short lengths, codes with close-to-optimal parameters can be constructed.
Furthermore, if efficient decoders are available for the ``component'' codes, the proposed code construction can be decoded efficiently up to $\LBtLC{\C}$, which in several cases exceeds half the minimum distance.

\textbf{Artifacts.} \ A Sagemath implementation of the proposed bounds is available at \url{github.com/sebastianbitzer/wHm}.

\section{Preliminaries}\label{sec:wHm}
\textbf{Notation.} \ 
Vectors are denoted with bold lowercase letters, and matrices with bold uppercase letters. 
We write $[m] \coloneqq \{1,\ldots,m\}$ and denote by $\F_q$ the finite field of $q$ elements.

\textbf{Hamming metric.} \ 
A $q$-ary linear code $\C$ of length $\nC{\C} = n$ and dimension $\kC{\C} = k$ is a $k$-dimensional linear subspace of $\F_q^{n}$.
Commonly, $\C$ is endowed with  the Hamming distance
\[
\dH{\6a}{\6b} \coloneqq \wH{\6a-\6b} \coloneqq \size{\{i \in [n]:a_i-b_i \neq 0\}}.
\]
The minimum distance of $\C$ is denoted as $\dHC{\C}$.
We associate with $\C$ an $\F_q$-linear encoder and a decoder, which may signal a decoding failure by returning $\perp$:
\begin{align*}
&\C.\enc{\cdot}\colon& \F_q^{k} &\to \C, & \6m &\mapsto \6c; \\ &
\C.\dec{\cdot}\colon& \F_q^{n} &\to \C\cup\{\perp\}, & \6v &\mapsto \hat{\6m}.
\end{align*}

\textbf{Block weight.} \ 
Let $\bn = (n_1, \ldots, n_m) \in \NN^m$ be a partition of $n = \sum_{\ell\in[m]} n_\ell$.
The \emph{block weight}~\cite{simonis} of $\6a = (\6a_1,\ldots,\6a_m) \in \F_q^{n}$ with $\6a_\ell \in \F_q^{n_\ell}$ is defined as the tuple
\[
\wN{\6a} = \left(\wH{\6a_1}, \ldots, \wH{\6a_m}\right) .
\]
For a single block, $m=1$ and $\bn = (n)$,  the block weight coincides with the Hamming weight. Additionally, the block weight induces the partial order
\[
\wN{\6a} \preceq \wN{\6b} \iff \wH{\6a_\ell} \leq \wH{\6b_\ell}\ \forall \ell \in [m].
\]

\textbf{Weighted-Hamming metric.} \ 
This work concerns the {weighted-Hamming weight}, a generalization of the Hamming weight that is coarser than the block weight.
Fix $\bLam = (\lambda_1,\ldots,\lambda_m) \in \NN^m$ and a vector of scaling coefficients with $\lambda_1\leq \ldots \leq \lambda_m$.
The
\textit{weighted-Hamming weight}
of $\6a = (\6a_1,\ldots,\6a_m)\in\F_q^n$, with $\6a_\ell \in \F_q^{n_\ell}$, is defined as
\[
\wL{\6a} \coloneqq \sum_{\ell\in[m]} \lambda_\ell \cdot \wH{\6a_\ell}.
\]
The weighted-Hamming ball and its difference set are defined as
\[
\bLw(t) = \{\6x\colon \wL{\6x} \leq t\}, \
\Delta\bLw(t) = \{\6x-\6y\colon \6x,\6y\in\bLw(t) \}.
\]

\textbf{Error correction.} \ 
As shown in~\cite{loeliger1994basic}, a linear code $\code$ uniquely corrects all wHm errors of weight at most $t$ if and only if 
\vspace{-0.1cm}\[
\code \cap \Delta\bLw(t) = \{\60\}. 
\]
The maximum $t$ for which this holds is the (wHm) error-correction capability $\tLC{\code}$, which may also be computed as
\begin{align*}
\tLC{\code} &\coloneqq \min_{\6c\in\code\setminus\{\60\}} \tL{\6c}\\ 
\text{with} \ \tL{\6c} &\coloneqq \min_{\6r} \max\{\wL{\6r}, \wL{\6c-\6r}\} - 1.
\end{align*}
The following lemma summarizes some relations between the introduced measures.

\begin{lemma}\label{lem:sub_support}
Let $\6a,\6b\in\F_q^n$ with $\wN{\6a} \preceq \wN{\6b}$.
Then,  $\wL{\6a} \leq \wL{\6b}$ and $\tL{\6a} \leq \tL{\6b}$.
Both $\wL{\6a}$ and $\tL{\6a}$ are fully determined by  $\wN{\6a}$.
\end{lemma}
\begin{proof}
Since $\wL{\6a} \leq \wL{\6b}$ is immediate, we focus on $\tL{\6a} \leq \tL{\6b}$.
Let $\rb = \arg \min_{\6r\in\F_q^n} \max(\wL{\6r},\wL{\6b-\6r}) - 1$.
Note that $r_{\6b,i} \in\{b_i, 0\}$.
Define $\ra \in \F_q^n$ as
\[
r_{\6a,i} = 
\begin{cases}
a_i & \text{if } r_{\6b,i} \neq 0,\\
0   & \text{else,}
\end{cases}
\]
which implies  $\wL{\rb} \geq \wL{\ra}$. 
Further, $\wL{\6b-\rb} \geq \wL{\6a-\ra}$ since $a_i - r_{\6a,i} = 0$ for all $i$ with $b_i - r_{\6b, i} = 0$.
Thus
\begin{align*}
\tL{\6b} &= \max\{\wL{\rb}, \wL{\6b-\rb}\} -1 \\
&\geq \max\{\wL{\ra}, \wL{\6a-\ra}\} - 1 \geq \tL{\6a}. \quad \qedhere 
\end{align*}
\end{proof} 

\section{Bounds}\label{sec:bounds}

This section discusses bounds on the maximum error-correction capability of a $k$-dimensional wHm code.

\textbf{Distance-based bounds.} \ 
In \cite{bitzer2024weighted}, various bounds for the minimum weighted-Hamming distance of a code were established.
A first observation is that these bounds on $\dLC{\code}$ can be used to obtain bounds on $\tLC{\code}$ via the relation
\begin{equation}\label{eq:bound_d_t}
\left\lfloor \frac{\dLC{\C}-1}{2} \right\rfloor \leq \tLC{\C} \leq \left\lfloor \frac{\dLC{\C}+ \lambda_m}{2} \right\rfloor - 1.
\end{equation}
The following lemma, a direct consequence of \Cref{eq:bound_d_t},
spells out the requirements on the minimum distance for achieving a given error-correction capability.
\begin{lemma}\label{lem:d_vs_t}
A code $\code$ with $\dLC{\code} = 2t+1$ is guaranteed to satisfy $\tLC{\C} \geq t$.
If $\dLC{\code} < 2(t+1)-\lambda_m$, then $\tLC{\C} < t$.
\end{lemma}

In the sequel, we directly bound $\tLC{\C}$ without detouring through the minimum distance.
For many parameter choices, these bounds improve upon those obtained via the distance-based argument; see \Cref{fig:t}.

\textbf{Packing and covering.} \ 
Upper and lower bounds on the cardinality of a code with error-correction capability $t$ follow from the sizes of the weighted-Hamming ball $\bLw(t)$ and its difference set $\Delta\bLw(t)$.
To compute these sizes, we enumerate the block weights attained in the respective sets, i.e., we let
\begin{align*}
\wN{\bLw(t)}&\coloneqq \left\{\wN{\6v}: \6v \in \bLw(t) \right\},\\
\wN{\Delta\bLw(t)}&\coloneqq\left\{\wN{\6v}: \6v \in \Delta\bLw(t) \right\}. 
\end{align*}
The cardinalities can be calculated as
\begin{align*}
\size{\bLw(t)} &= \sum_{\6i\in\wN{\bLw(t)}} \prod_{\ell\in[m]} \binom{n_\ell}{i_\ell} (q-1)^{i_\ell}, \\
\size{\Delta\bLw(t)} &= \sum_{\6i\in\wN{\Delta\bLw(t)}} \prod_{\ell\in[m]} \binom{n_\ell}{i_\ell} (q-1)^{i_\ell}.
\end{align*}

The following packing and covering bounds hold for general error sets~\cite{loeliger1994basic}; hence, no proof is given for the weighted-Hamming metric.

\begin{proposition}
Any code $\C\subseteq \F_q^n$ with $\tLC{\C} = t$ satisfies
\[
\kC{\C} \leq n - \log_q\size{\bLw(t)}.
\]
Moreover, there exists a code $\C$ such that $\tLC{\C} = t$ and 
\[
\kC{\C} \geq n - \log_q\size{\Delta\bLw(t)}.
\]
\end{proposition}

\textbf{Singleton.} \ 
Let $\C$ be a $k$-dimensional code.
Then, $\C$ admits a generator matrix of the form $\6G = (\6A,\6T)$ where $\6T\in \F_q^{k\times k}$ is a  lower triangular matrix.
The first row of $\6G$ provides a codeword $\6c\in \C\setminus\{\60\}$ such that $\6c_i = 0$ for all $i > n-k+1$. 
It follows that $\tLC{\C} \leq \tL{\6c}$.
Combining this argument with \Cref{lem:sub_support}, we directly obtain the following Singleton-like bound for arbitrary $k$-dimensional codes.

\begin{proposition}\label{prop:t_singleton}
Let $\cstar \in \F_q^n$ be such that $\cstar_i \neq 0 \iff i \leq n-k+1$.  
Then, any $k$-dimensional code $\C \subseteq \F_q^n$ satisfies 
\[
\tLC{\C} \leq \tL{\cstar}.
\]
\end{proposition}

\Cref{prop:t_singleton} implies that Maximum Distance Separable (MDS) codes are optimal with respect to the wHm error-correction capability.
This aligns with~\cite{bitzer2024weighted}, which shows that MDS codes also achieve the optimal wHm minimum distance.

\textbf{Linear programming.} \ 
Simonis derived in~\cite{simonis}  MacWilliams identities for the block-weight enumerator
\[
A_{\6i}(\C) = \size{\{\6c\in\C: \wN{\6c}=\6i\}},
\]
which in turn yield a Linear Programming (LP) bound.
Since the weighted-Hamming weight is coarser than the block weight, an LP bound for the weighted-Hamming metric follows; see \cite{bitzer2024weighted}.
By modifying the constraints and imposing $A_{\6i} = 0$ for all $\6i\in\wN{\Delta\bLw(t)}\setminus\{\60\}$, we obtain an refined LP bound for the error-correction capability.

\begin{proposition}
Denote as $K_{j_\ell}^{\mathrm{H}}(i_\ell)$ the Hamming-metric Krawtchouk coefficient for the $\ell$-th block, which is given as
\[
K_{j_\ell}^\mathrm{H}(i_\ell)=\sum_{s=0}^{j_\ell} \binom{n_\ell-i_\ell}{{j_\ell}-s}\binom{i_\ell}{s}(q-1)^{{j_\ell}-s}(-1)^s.
\]
Then, any code $\C$ with error-correction capability $\tLC{\C} = t$ satisfies $\kC{\C} \leq \sum_{\6w} A_{\6w}$ for $\sum_{\6w} A_{\6w}$ maximum under the constraints
\begin{equation*}
\begin{aligned}
A_{\6{0}} &=1, \\
A_{\6i} &\geq 0 & \forall \6i &\in  \wN{\F_q^n},  \\
A_{\6i} &=0 & \forall \6i &\in \wN{\Delta\bLw(t)}\setminus\{\60\}, \\
\sum_{\6i} \prod_{\ell\in[m]}    K_{j_\ell}^{\textnormal{H}}(i_\ell)A_{\6i} &\geq 0 & \forall \6j &\in \wN{\F_q^n}. 
\end{aligned}
\end{equation*}    
\end{proposition}

\textbf{Comparison.} \ 
\Cref{fig:t} illustrates upper and lower bounds on the maximum dimension of a code for varying error-correction capability $t$.
We set $m=2$ with $\6n = (7,7)$ and $\bLam = (1,2)$. 
The proposed bounds are compared against their distance-based counterparts from~\cite{bitzer2024weighted}, which rely on \Cref{lem:d_vs_t}.
For both $q = 2$ and $q=7$, the plot confirms that the proposed bounds are at least as tight as these counterparts, and for several parameter settings, they provide improvements.

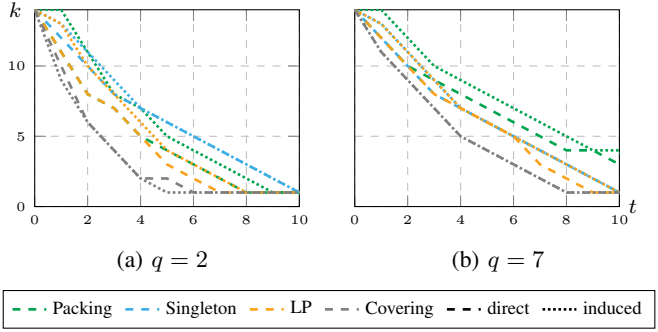
\begin{figure}[t]
\centering%
\begin{subfigure}[b]{0.5\columnwidth}
\centering
\begin{tikzpicture}[baseline] 
\begin{axis}[%
ymax = 14,
ymin = 0,
xmin = 0,
xmax = 10,
label style={font=\footnotesize},
ymajorgrids,
xmajorgrids,
grid style=dashed,
legend style={at={(1,1)},anchor=north east,font=\scriptsize,legend cell align=left, align=left, draw=black, legend image post style={xscale=0.75,yscale = 0.6,mark size=5pt}, },
ylabel near ticks,
xlabel near ticks,
width = 5.1cm,
height = 4.2cm,
legend columns=6,
legend entries={Packing\ , Singleton\ , LP\ , Covering\ , direct\ , induced},
legend to name={t_legend},
every axis y label/.style={font=\footnotesize,at={(current axis.north west)},above =0cm,left = 0.5mm},
ylabel={$k$},
ticklabel style = {font=\tiny},
]

\addplot [SPstyle] 
  table[row sep=crcr]{
0 14\\ 
1 11\\ 
2 8\\ 
3 7\\ 
4 5\\ 
5 4\\ 
6 3\\ 
7 2\\ 
8 1\\ 
9 1\\ 
10 1\\ 
11 0\\ 
};

\addplot [naiveSPstyle, forget plot] 
  table[row sep=crcr]{
0 14\\ 
1 14\\ 
2 11\\ 
3 8\\ 
4 7\\ 
5 5\\ 
6 4\\ 
7 3\\ 
8 2\\ 
9 1\\ 
10 1\\ 
11 1\\ 
12 0\\ 
};

\addplot [Singletonstyle] 
  table[row sep=crcr]{
0 14\\ 
1 12\\ 
2 10\\ 
3 8\\ 
4 7\\ 
5 6\\ 
6 5\\ 
7 4\\ 
8 3\\ 
9 2\\ 
10 1\\ 
11 0\\ 
};

\addplot [naiveSingletonstyle, forget plot] 
  table[row sep=crcr]{
0 14\\ 
1 13\\ 
2 11\\ 
3 9\\ 
4 7\\ 
5 6\\ 
6 5\\ 
7 4\\ 
8 3\\ 
9 2\\ 
10 1\\ 
11 0\\ 
};

\addplot [LPstyle] 
  table[row sep=crcr]{
0 14\\ 
1 11\\ 
2 8\\ 
3 7\\ 
4 5\\ 
5 3\\ 
6 2\\ 
7 1\\ 
8 1\\ 
9 1\\ 
10 1\\ 
11 0\\ 
};

\addplot [naiveLPstyle, forget plot] 
  table[row sep=crcr]{
0 14\\ 
1 13\\ 
2 10\\ 
3 8\\ 
4 6\\ 
5 4\\ 
6 3\\ 
7 2\\ 
8 1\\ 
9 1\\ 
10 1\\ 
11 0\\ 
};

\addplot [GVstyle] 
  table[row sep=crcr]{
0 14\\ 
1 10\\ 
2 6\\ 
3 4\\ 
4 2\\ 
5 2\\ 
6 1\\ 
7 1\\ 
8 1\\ 
9 1\\ 
10 1\\ 
11 0\\ 
};

\addplot [naiveGVstyle, forget plot] 
  table[row sep=crcr]{
0 14\\ 
1 9\\ 
2 6\\ 
3 4\\ 
4 2\\ 
5 1\\ 
6 1\\ 
7 1\\ 
8 1\\ 
9 1\\ 
10 1\\ 
11 0\\ 
};

\addplot [GVstyle, black] %
  table[row sep=crcr]{
-1 -1\\
};

\addplot [naiveGVstyle, black] %
  table[row sep=crcr]{
-1 -1\\
};

\end{axis}
\end{tikzpicture}
\caption{$q=2$}
\label{fig:t_q2}
\end{subfigure}
\begin{subfigure}[b]{0.48\columnwidth}
\centering
\begin{tikzpicture}[baseline] 
\begin{axis}[%
ymax = 14,
ymin = 0,
xmin = 0,
xmax = 10,
label style={font=\footnotesize},
xlabel={$t$},
ymajorgrids,
xmajorgrids,
grid style=dashed,
legend style={at={(1,1)},anchor=north east,font=\scriptsize,legend cell align=left, align=left, draw=black},
xlabel near ticks,
width = 5.1cm,
height = 4.2cm,
ymajorticks=false,
ticklabel style = {font=\tiny},
every axis x label/.style={font=\footnotesize,at={(current axis.south east)},below =0cm,right = 0mm},
]

\addplot [SPstyle] 
  table[row sep=crcr]{
0 14\\ 
1 12\\ 
2 10\\ 
3 9\\ 
4 8\\ 
5 7\\ 
6 6\\ 
7 5\\ 
8 4\\ 
9 4\\ 
10 3\\ 
11 2\\ 
12 2\\ 
13 1\\ 
14 1\\ 
15 1\\ 
16 0\\ 
};

\addplot [naiveSPstyle, forget plot] 
  table[row sep=crcr]{
0 14\\ 
1 14\\ 
2 12\\ 
3 10\\ 
4 9\\ 
5 8\\ 
6 7\\ 
7 6\\ 
8 5\\ 
9 4\\ 
10 4\\ 
11 0\\ 
};

\addplot [Singletonstyle, dash phase = 1.6pt]
  table[row sep=crcr]{
0 14\\ 
1 12\\ 
2 10\\ 
3 8\\ 
4 7\\ 
5 6\\ 
6 5\\ 
7 4\\ 
8 3\\ 
9 2\\ 
10 1\\ 
11 0\\ 
};

\addplot [naiveSingletonstyle,forget plot, dash phase = 0.5pt]
  table[row sep=crcr]{
0 14\\ 
1 13\\ 
2 11\\ 
3 9\\ 
4 7\\ 
5 6\\ 
6 5\\ 
7 4\\ 
8 3\\ 
9 2\\ 
10 1\\ 
11 0\\ 
};

\addplot [ LPstyle]
  table[row sep=crcr]{
0 14\\ 
1 12\\ 
2 10\\ 
3 8\\ 
4 7\\ 
5 6\\ 
6 5\\ 
7 3\\ 
8 2\\ 
9 1\\ 
10 1\\ 
};

\addplot [ naiveLPstyle, forget plot]
  table[row sep=crcr]{
0 14\\ 
1 13\\ 
2 11\\ 
3 9\\ 
4 7\\ 
5 6\\ 
6 5\\ 
7 4\\ 
8 3\\ 
9 2\\ 
10 1\\ 
11 0\\ 
};

\addplot [GVstyle] 
  table[row sep=crcr]{
0 14\\ 
1 11\\ 
2 9\\ 
3 7\\ 
4 5\\ 
5 4\\ 
6 3\\ 
7 2\\ 
8 1\\ 
9 1\\ 
10 1\\ 
11 0\\ 
};

\addplot [naiveGVstyle, forget plot] 
  table[row sep=crcr]{
0 14\\ 
1 11\\ 
2 9\\ 
3 7\\ 
4 5\\ 
5 4\\ 
6 3\\ 
7 2\\ 
8 1\\ 
9 1\\ 
10 1\\ 
11 0\\ 
};

\end{axis}
\end{tikzpicture}
\caption{$q=7$}         
\label{fig:t_q7}
\end{subfigure}\hfill
\vspace{0.2cm}
\pgfplotslegendfromname{t_legend}
\caption{Bounds on the maximum dimension $k$ of a linear code with error-correcion capability $t$. 
Block lengths $\6n=(7,\,7)$ and scaling factors $\bm{\lambda} = (1,\,2)$.}
\label{fig:t}
\end{figure}

\section{Code Construction}\label{sec:codes}

This section proposes a flexible code construction for the weighted-Hamming metric.
Generalized Concatenated Codes (GCC), originally introduced by Zinov'ev~\cite{zinov1976generalized}, have been successfully applied to a wide range of error models~\cite{zinoviev2021generalized,thiers2022code,bitzer2025bounds}.
Building on this framework, we consider a generalization that uses polyalphabetic outer codes and distinct inner codes for each block (similar to \cite{dettmar1995modified}).
The resulting flexibility allows the construction to be tailored to given scaling coefficients $\lambda_1,\ldots,\lambda_m$. Moreover, we show that a modified version of the original decoder achieves decoding up to a lower bound on the true error-correction capability.

\subsection{Polyalphabetic Codes}

Classical Hamming-metric codes are defined over $\F_q^n$, i.e., all coordinates take values in the same field.
In contrast, polyalphabetic (or mixed) codes are defined over the ambient space $\AmbSpace = \mathcal{Q}_1\dirprod\dotsb\dirprod\mathcal{Q}_n$ for $n$ alphabets $\mathcal{Q}_1,\ldots,\mathcal{Q}_n$; see~\cite{sidorenko2005polyalphabetic}. 

In this work, we use $\mathcal{Q}_\ell = \F_{q}^{\extdeg_\ell}$ for $\ell\in[n]$.
More specifically, we use polyalphabetic codes $\codeA$ that are $\F_q$-linear subspaces of $\AmbSpace = \F_{q}^{\extdeg_1}\dirprod\dotsb\dirprod\F_{q}^{\extdeg_n}$, equipped with the Hamming distance, where the elements of $\F_q^{\extdeg_i}$ are used as symbols.
As in the monoalphabetic case, we denote by $n = \nC{\codeA}$ the length, by $\kC{\codeA}$ the $\F_q$-dimension, and by 
\[
\dHC{\codeA} = \min_{\6a\in\codeA\setminus\{\60\}} \size{\{i\in[n]:\6a_i\in \F_q^{\extdeg_i}\setminus\{\60\}\}}
\]
the minimum Hamming distance.

The following construction recalls the derivation of polyalphabetic codes from monoalphabetic mother codes due to \cite{sidorenko2005polyalphabetic}.

\begin{construction}\label{const:poly}
We construct $\codeA \subset \AmbSpace = \F^{\extdeg_1}_{q} \dirprod \dotsb \dirprod \F^{\extdeg_n}_{q}$.
Without loss of generality, we assume that $\extdeg_1 \leq \ldots \leq \extdeg_n$.
Let $\C \subseteq \F^n_{q^{\extdeg_k}}$ have $\F_{q^{\extdeg_k}}$-dimension $k$ equipped with a systematic encoder, i.e., $\C.\enc{\6m}_{[k]} = \6m$.
Denote by $\Gamma$ an $\F_q$-linear mapping from $\F_{q^{\extdeg_k}}$ to $\F_{q}^{\extdeg_k}$.
We write \[
\Gamma(\C) = \{(\Gamma(c_1),\ldots,\Gamma(c_n)): (c_1,\ldots,c_n)\in C\}.
\]
$\Gamma(\C)$ contains a subcode $\Gamma(\C)'$ such that all $(\6c_1,\ldots\6c_n)\in\Gamma(\C)'$ satisfy
$\6c_i \in \F_q^{\extdeg_i}\dirprod\{0\}^{\extdeg_k-\extdeg_i}$ for any $i\in[k]$.
Then $\codeA$ is obtained as
\[
\codeA \coloneqq \left\{ (\6a_1,\ldots,\6a_n)\in\AmbSpace:\hspace{0.1cm}
\begin{aligned}
&(\6c_1,\ldots,\6c_n) \in \Gamma(\C)', \\
&\6c_i = (\6a_i, \60)\text{ for }i \leq k, \\
&\6a_i = (\6c_i,\60)\text{ for } i > k
\end{aligned}
\right\}.
\]
To embed $\Gamma(\C)'$ into $\AmbSpace$, the zero-coordinates are punctured in the first $k$ positions and zero padding is applied to the remaining $n-k$ positions.
Then $\codeA\in\AmbSpace$ has $\F_q$-dimension $\kC{\codeA} = \sum_{\ell\in[k]} \extdeg_\ell$ and satisfies $\dHC{\codeA} \geq \dHC{\C}$.
\end{construction}

\begin{example}\label{ex:poly}
    Let $\C \subset \F_{q^2}^3$ be the parity-check code, i.e., $\C$ has $\F_{q^2}$-dimension $\kC{\C} = 2$ and 
    $\dHC{\C} = 2$.
    Applying \Cref{const:poly}, we obtain a code $\codeA\subset \F_q \dirprod \F^2_{q}\dirprod\F^3_{q}$ with $\F_q$-dimension $\kC{\codeA} = 3$ and minimum Hamming distance $\dHC{\codeA} = 2$. 
\end{example}

For further details on polyalphabetic codes, we refer the interested reader to \cite{sidorenko2005polyalphabetic,yehezkeally2023bounds} and the references therein. 

\subsection{Generalized Concatenated Codes}

Classical GCCs combine several outer codes with a single nested sequence of inner codes; see \cite{bossert1999channel} for a comprehensive introduction.
To construct GCCs for the weighted-Hamming metric with block sizes $(n_1,\,\ldots,\,n_m)$ and scaling coefficients $(\lambda_1,\,\ldots,\,\lambda_m)$, we consider, for each block $\ell\in[m]$, a nested sequences of inner codes 
\[
\{\60\} = \codeB_{s+1, \ell} \subseteq  \codeB_{s, \ell}\subseteq \dotsb \subseteq \codeB_{1, \ell} \subseteq \F_q^{n_\ell},
\]
where $s \leq m$ is the number of levels.
By abuse of notation, $\codeB_{j,\ell}/\codeB_{j+1,\ell} \subset \F_q^{n_\ell}$ denotes the (not necessarily unique) linear code of dimension $\kC{\codeB_{j,\ell}}-\kC{\codeB_{j+1,\ell}}$ such that
\[
\codeB_{j,\ell} = \{\6c_1 + \6c_2 \mid \6c_1 \in\codeB_{j+1,\ell}, \6c_2 \in \codeB_{j,\ell}/\codeB_{j+1,\ell} \}. 
\]

Before using this notation to introduce the proposed code construction, we illustrate it with a simple example.
\begin{example}
Let $\codeB_{j,\ell} = \F_2^3$ be generated by 
$\6g_1 = \begin{pmatrix}
1 & 1 & 1\\
\end{pmatrix}$,
$\6g_2 = \begin{pmatrix}
0 & 1 & 0\\
\end{pmatrix}$, and
$\6g_3 = \begin{pmatrix}
0 & 0 & 1
\end{pmatrix}$.
We select the subcode $\codeB_{j+1,\ell} = \langle\6g_1\rangle$ which is the repetition code.
Then, we may select $\codeB_{j,\ell}/\codeB_{j+1,\ell} = \langle \6g_2,\6g_3\rangle$.
\end{example}

\begin{construction}\label{const:gcc} 
For $\ell \in[m]$, let a nested sequence of inner codes $\{\60\} = \codeB_{s+1, \ell} \subseteq  \codeB_{s, \ell}\subseteq \dotsb \subseteq \codeB_{1, \ell} \subseteq \F_q^{n_\ell}$ be given.
For $j \in [s]$, let $\codeA_j \subseteq \F_{q}^{\extdeg_{j,1}} \times \cdots \times \F_{q}^{\extdeg_{j,m}}$ with $\extdeg_{j,\ell} = \kC{\codeB_{j,\ell}}-\kC{\codeB_{j+1,\ell}}$.
Each $\codeA_j$ is $\F_q$-linear with $\F_q$-dimension $\kC{\codeA_j}$ and Hamming distance $\dHC{\codeA_j}$.
The \emph{generalized concatenated code} $\C$ is constructed as {\small
\[
\C = \left\{\!\Big( \sum_{j\in[s]} (\codeB_{j,\ell} / \codeB_{j+1,\ell} ).\enc{a_{j,\ell}}\Big)_{\ell\in[m]}\mid \6a_j \in \codeA_j\ \forall j \in [s]\right\}.
\]}Then, $\C$ has length $\nC{\C} = \sum_{\ell\in[m]} n_\ell$, dimension $\kC{\C} = \sum_{j\in[s]} \kC{\codeA_j}$, and minimum wHm distance 
\[
\dLC{\C} \geq \LBdLC{\C} \coloneqq \min_{j\in [s]} \sum_{\ell\in[\dHC{\codeA_j}]} \lambda_{(\ell)}\cdot \dHC{\codeB_{j,(\ell)}},
\]
where $\lambda_{(1)}\cdot \dHC{\codeB_{j,(1)}} \leq \dots \leq \lambda_{(m)}\cdot \dHC{\codeB_{j,(m)}}$ denotes the ordered sequence of weighted inner-code distances at level~$j$.
\end{construction}

\begin{proof}
The statements on the length and dimension follow directly from classical arguments.
Let $\6a_j \in\codeA_j$ denote the outer codeword encoded at level $j\in[s]$.
Let $\jstar$ be such that $\6a_j = \60$ for $j< \jstar$ and that $\6a_{\jstar} \neq \60$; let $\6a_j$ be arbitrary for $j > \jstar$. 
Since $\codeA_j$ has minimum Hamming distance $\dHC{\codeA_{\jstar}}$, $\6a_{\jstar}$ has at least $\dHC{\codeA_{\jstar}}$ non-zero coefficients. By the $\F_q$-linearity of the encoding, each such symbol results in a non-zero codeword of $\codeB_{\ell,\jstar}$.
As $\codeB_{\jstar,\ell}$ has minimum Hamming distance $\dHC{\codeB_{\jstar,\ell}}$, each non-zero codeword contributes at least $\lambda_\ell \cdot \dHC{\codeB_{\jstar,\ell}}$ to the weighted-Hamming weight.
%
Ordering the quantities $\lambda_\ell \cdot \dHC{\codeB_{j^\star,\ell}}$  as $\lambda_{(1)} \cdot \dHC{\codeB_{j^\star,(1)}} \le \dots \le \lambda_{(m)} \cdot \dHC{\codeB_{j^\star,(m)}}$, the total contribution of the $\dHC{\codeA_{j^\star}}$ non-zero positions is lower bounded by
\[
\sum_{\ell\in[\dHC{\codeA_{\jstar}}]}
\lambda_{(\ell)} \, \dHC{\codeB_{\jstar,(\ell)}}.
\]
Minimizing over $\jstar\in[s]$ yields the stated lower bound.
\end{proof}

For illustration, we first consider the case of $s=1$, which corresponds to a concatenated code.
If, in addition, the outer code is chosen to have rate one, the construction further specializes to independent coding across the blocks.

\begin{example}
Let $\6n = (3,3,3)$ and $\bLam = (1,2,3)$.
We choose $s = 1$ and an outer code $\codeA_1\subset \F_q\times \F^2_{q}\times\F^3_{q}$ constructed as in $\Cref{ex:poly}$, with $\kC{\codeA_1} = 3$ and $\dHC{\codeA_1} = 2$.
As inner codes, we take $\codeB_{1,1}, \codeB_{1,2}, \codeB_{1,3} \subset \F_q^3$ to be a repetition code, a single parity-check code, and a rate-one code, respectively.
The constructed GCC $\C$ has dimension $\kC{\C} = 3$ and minimum weighted Hamming distance $\dLC{\C} \geq 3+3 = 6$, because the weighted inner-code distances are
\[
\lambda_1\cdot\dHC{\codeB_{1,1}} = 3, \quad
\lambda_2\cdot\dHC{\codeB_{1,2}} = 4, \quad
\lambda_3\cdot\dHC{\codeB_{1,3}} = 3. 
\]
\end{example}


A special case of \Cref{const:gcc} 
is our previous construction~\cite{bitzer2024weighted}, which achieves $\LBdLC{\C} = 5$ for $\bLam = (1,2)$.

\begin{remark}\label{rem:old_code}
Let $m=2$ and $\bLam=(1,2)$.
To construct a GCC $\C$ with $\dLC{\C}= 5$, we pick $s=2$ and inner codes such that
\begin{align*}
\codeB_{2,1} &\subseteq \codeB_{1,1} \subseteq \F_q^{n_1},
& \dHC{\codeB_{2,1}} &= 5,
& \dHC{\codeB_{1,1}} &= 3,\\
\codeB_{2,2} &\subseteq \codeB_{1,2} = \F_q^{n_2},
& \dHC{\codeB_{2,2}} &= 3,
& \dHC{\codeB_{1,2}} &= 1.
\end{align*}
Let $\extdeg_{j,\ell} = \kC{\codeB_{j,\ell}}-\kC{\codeB_{j+1,\ell}}$ for $\ell\in[2]$ and $j\in[s]$.
We select outer codes
\begin{align*}
\codeA_1 &\subseteq \F^{\extdeg_{1,1}}_{q} \times \F^{\extdeg_{1,2}}_{q},
& \dHC{\codeA_1} &= 2,\\
\codeA_2 &= \F^{\extdeg_{2,1}}_{q} \times \F^{\extdeg_{2,2}}_{q},
& \dHC{\codeA_2} &= 1.
\end{align*}
The resulting generalized concatenated code has dimension 
\[
\min\{\extdeg_{1,1}, \extdeg_{1,2}\} + \kC{\codeB_{2,1}} + \kC{\codeB_{2,2}}
\]
and minimum weighted-Hamming distance equal to~$5$.
\end{remark}

One can check that the codes constructed in \Cref{rem:old_code} have error-correction capability $\tLC{\C} = 2$, which corresponds to the lowest possible value according to \Cref{lem:d_vs_t}.
In general, however, the error-correction capability of the codes in \Cref{const:gcc} \textit{can exceed} half the minimum distance.
The following theorem provides a lower bound on the true error-correction capability.

\begin{theorem}\label{thm:gcc_tau_conc}%
Let $\C$ be obtained from \Cref{const:gcc} with inner codes $\codeB_{j,\ell}$ and outer codes $\codeA_j$ for $j\in[s]$ and $\ell\in [m]$. 
For each $j\in[s]$ and $\ell\in[m]$, fix a vector $\6v_{j,\ell}\in\F_q^{n_\ell}$ with $\wH{\6v_{j,\ell}} = \dHC{\codeB_{j,\ell}}$.
Define $\Vset = \bigcup_{j\in[s]}\Vset_j$ with
\[
\Vset_j \coloneqq
\left\{
(u_1 \6v_{\!j,1}, \ldots, u_m \6v_{\!j,m})
\middle|
\6u \in \{0,1\}^m\!\!,\wH{\6u} = \dHC{\codeA_j}
\right\}.
\]
Then, the error-correction capability of $\C$ satisfies
\[
\tLC{\C} \geq \LBtLC{\C} \coloneqq \min_{\6v\in \Vset} \tL{\6v}.
\]
\end{theorem}
\begin{proof}
Let $\6c \in \C\setminus\{\60\}$.
By \Cref{const:gcc}, $\6c$ is of the form 
\[
\6c = \Big( \sum_{j\in[s]} (\codeB_{j,\ell} / \codeB_{j+1,\ell} ).\enc{\6a_{j,\ell}}\Big)_{\ell\in[m]}.
\]
Let $\6a_j = \60$ for all $j<\jstar$ and $\6a_{\jstar}\neq \60$.
Then $\6c$ has at least $\dHC{\codeA_{\jstar}}$ blocks with non-zero codewords from the inner codes $\codeB_{\jstar,\ell}$.    
Consequently there exists a vector $\6v\in\Vset_{\jstar}$ such that
$\wN{\6v} \preceq \wN{\6c}$.
By \Cref{lem:sub_support}, this implies $\tL{\6v} \leq \tL{\6c}$, and the claim follows.
\end{proof}

\Cref{thm:gcc_tau_conc} requires evaluating $\tL{\6v}$ for at most $s\cdot2^m$ vectors $\6v$.
The obtained lower bound on the error-correction capability is non-trivial in the sense that there exist concatenated codes for which the lower bound exceeds half the true minimum distance.
This is illustrated in the following example.

\begin{example}\label{ex:tau}
Let $\6n = (3,3)$ and $\bLam = (1, 2)$.
$\C$ is obtained by concatenating $\codeA_1 = \F_2 \times \F^3_{2}$ with a repetition code $\codeB_{1,1}\subset\F_2^3$ and a rate-one code $\codeB_{1,2}\subset\F_2^3$, so that $\dHC{\codeB_{1,1}} = 3$ and $\dHC{\codeB_{1,2}} = 1$.
Then $\C$ has minimum weighted-Hamming distance $\dHC{\C} = 2$, while Thm.~\ref{thm:gcc_tau_conc} gives $\tLC{\C} \geq \LBtLC{\C} = 1$.
\end{example}

\subsection{Weighted-Hamming Metric Decoder}

Classical GCCs are decoded up to half the minimum distance using the Blokh-Zyablov-Zinov'ev algorithm~\cite{zinov1981generalized}; see also \cite{bossert1988decoding,bossert1999channel} for a detailed treatment.
For \Cref{const:gcc}, each block $\ell\in[m]$ has its own scaling coefficient $\lambda_\ell$ and inner codes $\codeB_{j,\ell}$.
Accordingly, we adjust the reliabilities passed from the inner to the outer decoder following \cite{dettmar1995modified}.
\Cref{alg:dec} summarizes the resulting decoder.

\begin{algorithm}[h]
\caption{Decoding Algorithm $\C.\dec{\cdot}$}\label{alg:dec}


\begin{algorithmic}
\Statex 
\Statex \makebox[\widthof{\textbf{Output:}}]{\textbf{Input:}\hfill} noisy sequence $\6r = \6c + \6e = (\6r_1,\ldots,\6r_m)$
\Statex \textbf{Output:} estimated codeword $\hat{\6c} = \sum_{j\in[s]} \hat{\6c}_j$
\Statex \rule{\linewidth}{0.5pt}
\State $\hat{\6c} \gets \60$, $\6r_1 \gets \6r$
\For{$j\in[s]$}
    \For{$\ell\in[m]$}
        \State ${\6b}'_{j,\ell} \gets \codeB_{j,\ell}.\dec{\6r_{j,\ell}}$ \hfill \textit{// Hamming metric decoder}
        \vspace{0.2cm}
        \If{$\6b'_{j,\ell} = {\perp}$}\hfill\textit{// decoding failure}
            \State ${\6a}'_{j,\ell} \gets \60$, $\alpha_{j,\ell} \gets 0$\hfill\textit{// outer symbol,  reliability}
        \Else
            \State $\6a'_{j,\ell} \gets (\codeB_{j,\ell} / \codeB_{j+1,\ell} ).\encinv{\6b'_{j,\ell}}$
            \State $\alpha_\ell \gets \lambda_\ell \cdot \max\{0,\dHC{\codeB_{j,\ell}} - 2 d_\mathrm{H}(\6r_{j,\ell}, \6b'_{j,\ell})\}$
        \EndIf
    \EndFor
    \vspace{0.2cm}
    \State $\hat{\6a}_j \gets \codeA_j.\GMDdec{\6a', \bm\alpha}$\hfill\textit{// GMD decoder} 
    \State$\hat{\6c}_j \gets \left( (\codeB_{\ell,j} / \codeB_{\ell,j+1} ).\enc{\hat{\6a}_{j,\ell}}\right)_{\ell\in[m]}$ \hfill\textit{// re-encode} 
    \State$\hat{\6c} \gets \hat{\6c} + \hat{\6c}_j$, $\6r_{j+1} \gets \6r_j - \hat{\6c}_j$  \hfill\textit{// successive cancellation} 
\EndFor
\State \Return $\hat{\6c}$
\end{algorithmic}
\end{algorithm}

To decode $\6r = \6c + \6e$, the decoder proceeds level by level in a successive-cancellation-like manner.
Let $\6r_j = (\6r_{j,1},\ldots,\6r_{j,m})$ denote the residual after removing contributions from the previously decoded levels $1,\ldots,j-1$.
At level $j$,  for each $\ell\in[m]$, decoding $\codeB_{j,\ell}$ produces an estimated inner codeword ${\6b}'_{j,\ell}$.
If a decoding output is available (${\6b}'_{j,\ell} \neq {\perp}$), the corresponding outer symbol is obtained as an $\6a'_{j,\ell} \in \F^{\extdeg_{j,\ell}}_{q}$ such that there exists ${\6b}'_{j+1,\ell} \in \codeB_{j+1,\ell}$ with
\[
(\codeB_{j,\ell}/\codeB_{j+1,\ell}).\enc{a'_\ell} + {\6b}'_{j+1,\ell} =  {\6b}'_{j,\ell}.
\]
In \Cref{alg:dec}, this is denoted as $(\codeB_{\ell,j} / \codeB_{\ell,j+1} ).\encinv{{\6b}'_\ell}$.
The symbol reliability is then computed as 
\[
\alpha_\ell = \lambda_\ell \cdot \max\{0,\,\dHC{\codeB_{j,\ell}}-2\dH{\6r_{j,\ell}}{{\6b}'_{j,\ell}}\}.
\]
If decoding fails (${\6b}'_\ell = {\perp}$), we set $a_\ell = 0$  and $\alpha_\ell = 0$. 
The estimates $(a'_{j,\ell}, \alpha_{j,\ell})_{\ell\in[m]}$ are passed to a Generalized Minimum Distance (GMD) decoder for the corresponding outer code~\cite{forney1966generalized}.
\noindent The result is then re-encoded and canceled from $\6r_j$ before proceeding to the next level.

\begin{theorem}
\label{thm:dec}
Let $\C$ be constructed according to \Cref{const:gcc}.
For $j\in[s]$ and $\ell\in[m]$, assume that $\codeB_{j,\ell}$ is decoded using a Hamming-metric bounded-minimum-distance decoder and that~$\codeA_j$ is decoded using a GMD decoder~\cite{forney1966generalized}. 
Then \Cref{alg:dec} correctly decodes $\C$ for all errors up to $\LBtLC{\C}$.
\end{theorem}
\begin{proof}
Fix an error $\6e$ with $\wL{\e} \leq \LBtLC{\C}$ and an arbitrary level $j\in[s]$.
We show that the GMD decoder for $\codeA_j$ outputs the correct outer codeword.

Let $\wN{\6e} = (w_1,\ldots,w_m)$ denote the block weights of $\6e$.
Without loss of generality, we may assume that $w_\ell \leq \dHC{\codeB_{j,\ell}}$, since for any error $\6e$ violating this condition, there exists an error of smaller weight that yields the same decoding outcome.
Let $\Eset = \{\ell \mid \6b_{j,\ell} \neq \6b'_{j,\ell}\}$, $\Gset = [m]\setminus\Eset$, and $\size{\Eset} = f$.
Let $\Gsetstar \subset \Gset$ be the $\max\{0,\dHC{\codeA_j}-f\}$ least reliable correctly decoded blocks.
By \cite{taipale1991improvement}, GMD decoding succeeds if
\begin{equation}\label{eq:Taipale}
\sum_{\ell\in\Gsetstar} \alpha_\ell > \sum_{\ell\in\Eset} \alpha_\ell,
\end{equation}
independently of the chosen reliability measure.
Let us define 
\[
\Vset'_{\!j} = \left\{\6v\in\F_q^n\colon
\size{\{\ell\in[m]\colon \wH{\6v_{\!\ell}} \geq \dHC{\codeB_{j,\ell}}\}} \geq \dHC{\codeA_j} \right\}.
\]
Compare $\Vset'_j$ with $\Vset_j$ defined in Thm.~\ref{thm:gcc_tau_conc}:  for all $\6e$ with $\wL{\6e} \leq \LBtLC{\C}$ and all $\6v\in\Vset'_j$, it holds that $\wL{\6e} < \dL{\6e}{\6v}$. 
Now, assume that \Cref{eq:Taipale} is violated.
We construct $\6v \in \Vset'_j$ such that $\wL{\6e} \geq \dL{\6e}{\6v}$, implying $\wL{\6e} > \LBtLC{\C}$. 

Let $\Uset = \Eset\cup \Gsetstar$, which implies $\size{\Uset} \geq \dHC{\codeA_j}$.
For $\ell \in \Uset$, we pick $\6v_\ell$ such that $\wH{\6v_\ell} = \dHC{\codeB_{j,\ell}}$ and $v_{\ell, i} = e_{\ell,i}$ for all $i$ with $\6e_{\ell,i} \neq 0$.
For $\ell \notin \Uset$, we select $\6v_\ell = \60$.
Then,  $\6v \in\Vset'_j$ and, using   the definition of the reliabilities $\alpha_\ell$, we obtain
\[
\dL{\6v_\ell}{\6e_\ell} = \lambda_\ell (\dHC{\codeB_{j,\ell}} - w_\ell) \leq 
\begin{cases}
\lambda_\ell w_\ell -\alpha_\ell & \text{for } \ell \in \Eset,\\
\lambda_\ell w_\ell +\alpha_\ell & \text{for }\ell \in \Gsetstar.
\end{cases}
\vspace*{0.13in} 
\]

Summing over all blocks yields 
\begin{align*}
\dL{\6e}{\6v} 
&\leq \sum_{\ell \in \Gsetstar} (\alpha_\ell + w_\ell \lambda_\ell) 
+ \sum_{\ell \in \Eset} (\lambda_\ell w_\ell - \alpha_\ell)
+ \sum_{\ell \notin \Uset} \lambda_\ell w_\ell\\
&= \wL{\6e}  +\sum_{\ell \in \Gsetstar} \alpha_\ell - \sum_{\ell \in \Eset} \alpha_\ell \leq \wL{\6e}.
\end{align*}
Thus any error violating \Cref{eq:Taipale} satisfies $\wL{\6e} > \LBtLC{\C}$.
\end{proof}

\subsection{Achievable Parameters}

\begin{figure}[t]
\centering%
\begin{subfigure}[b]{0.5\columnwidth}
\centering
\begin{tikzpicture}[baseline] 
\begin{axis}[%
ymax = 21,
ymin = 0,
xmin = 0,
xmax = 42,
label style={font=\footnotesize},
ymajorgrids,
xmajorgrids,
grid style=dashed,
legend style={at={(1,1)},anchor=north east,font=\scriptsize,legend cell align=left, align=left, draw=black, legend image post style={xscale=0.75,yscale = 0.6,mark size=4pt}, },
ylabel near ticks,
xlabel near ticks,
width = 5.1cm,
legend columns=6,
legend entries={Packing\ , Singleton\ , LP\ , Covering\ , \Cref{const:gcc}},
legend to name={d_GCC_m3_legend},
every axis y label/.style={font=\footnotesize,at={(current axis.north west)},above =1mm,left = 0.5mm},
ylabel={$k$},
ticklabel style = {font=\tiny},
]

\addplot [SPstyle] 
  table[row sep=crcr]{
1 21\\ 
2 21\\ 
3 18\\ 
4 18\\ 
5 15\\ 
6 15\\ 
7 14\\ 
8 14\\ 
9 12\\ 
10 12\\ 
11 11\\ 
12 11\\ 
13 9\\ 
14 9\\ 
15 8\\ 
16 8\\ 
17 7\\ 
18 7\\ 
19 6\\ 
20 6\\ 
21 5\\ 
22 5\\ 
23 5\\ 
24 5\\ 
25 4\\ 
26 4\\ 
27 3\\ 
28 3\\ 
29 3\\ 
30 3\\ 
31 2\\ 
32 2\\ 
33 2\\ 
34 2\\ 
35 2\\ 
36 2\\ 
37 1\\ 
38 1\\ 
39 1\\ 
40 1\\ 
41 1\\ 
42 1\\ 
};

\addplot [Singletonstyle] 
  table[row sep=crcr]{
1 21\\ 
2 20\\ 
3 19\\ 
4 18\\ 
5 17\\ 
6 16\\ 
7 15\\ 
8 14\\ 
9 14\\ 
10 13\\ 
11 13\\ 
12 12\\ 
13 12\\ 
14 11\\ 
15 11\\ 
16 10\\ 
17 10\\ 
18 9\\ 
19 9\\ 
20 8\\ 
21 8\\ 
22 7\\ 
23 7\\ 
24 7\\ 
25 6\\ 
26 6\\ 
27 6\\ 
28 5\\ 
29 5\\ 
30 5\\ 
31 4\\ 
32 4\\ 
33 4\\ 
34 3\\ 
35 3\\ 
36 3\\ 
37 2\\ 
38 2\\ 
39 2\\ 
40 1\\ 
41 1\\ 
42 1\\ 
};

\addplot [LPstyle] 
  table[row sep=crcr]{
1 21\\ 
2 20\\ 
3 18\\ 
4 17\\ 
5 15\\ 
6 15\\ 
7 14\\ 
8 13\\ 
9 12\\ 
10 11\\ 
11 10\\ 
12 10\\ 
13 9\\ 
14 8\\ 
15 7\\ 
16 7\\ 
17 6\\ 
18 6\\ 
19 5\\ 
20 5\\ 
21 4\\ 
22 4\\ 
23 3\\ 
24 3\\ 
25 2\\ 
26 2\\ 
27 2\\ 
28 2\\ 
29 1\\ 
30 1\\ 
31 1\\ 
32 1\\ 
33 1\\ 
34 1\\ 
35 1\\ 
36 1\\ 
37 1\\ 
38 1\\ 
39 1\\ 
40 1\\ 
41 1\\ 
42 1\\ 
};

\addplot [GVstyle] 
  table[row sep=crcr]{
1 21\\ 
2 18\\ 
3 16\\ 
4 15\\ 
5 13\\ 
6 12\\ 
7 10\\ 
8 9\\ 
9 8\\ 
10 7\\ 
11 6\\ 
12 6\\ 
13 5\\ 
14 4\\ 
15 4\\ 
16 3\\ 
17 3\\ 
18 3\\ 
19 2\\ 
20 2\\ 
21 2\\ 
22 1\\ 
23 1\\ 
24 1\\ 
25 1\\ 
26 1\\ 
27 1\\ 
28 1\\ 
29 1\\ 
30 1\\ 
31 1\\ 
32 1\\ 
33 1\\ 
34 1\\ 
35 1\\ 
36 1\\ 
37 1\\ 
38 1\\ 
39 1\\ 
40 1\\ 
41 1\\ 
42 1\\ 
};

\addplot[newCstyle]
  table[row sep=crcr]{
1 21\\ 
2 20\\ 
3 18\\ 
4 17\\ 
5 15\\ 
6 14\\ 
7 13\\ 
8 12\\ 
10 10\\ 
12 9\\ 
14 7\\ 
16 6\\ 
18 5\\ 
21 4\\ 
24 3\\ 
26 2\\ 
42 1\\ 
}; 

\end{axis}
\end{tikzpicture}
\end{subfigure}
\begin{subfigure}[b]{0.48\columnwidth}
\centering
\begin{tikzpicture}[baseline] 
\begin{axis}[%
ymax = 21,
ymin = 0,
xmin = 1,
xmax = 42,
label style={font=\footnotesize},
xlabel={$d$},
ymajorgrids,
xmajorgrids,
grid style=dashed,
legend style={at={(1,1)},anchor=north east,font=\scriptsize,legend cell align=left, align=left, draw=black},
xlabel near ticks,
width = 5.1cm,
ymajorticks=false,
ticklabel style = {font=\tiny},
every axis x label/.style={font=\footnotesize,at={(current axis.south east)},below =0cm,right = 0mm},
]

\addplot [SPstyle] 
  table[row sep=crcr]{
1 21\\ 
2 21\\ 
3 19\\ 
4 19\\ 
5 17\\ 
6 17\\ 
7 16\\ 
8 16\\ 
9 15\\ 
10 15\\ 
11 14\\ 
12 14\\ 
13 13\\ 
14 13\\ 
15 12\\ 
16 12\\ 
17 11\\ 
18 11\\ 
19 10\\ 
20 10\\ 
21 10\\ 
22 10\\ 
23 9\\ 
24 9\\ 
25 8\\ 
26 8\\ 
27 8\\ 
28 8\\ 
29 7\\ 
30 7\\ 
31 7\\ 
32 7\\ 
33 6\\ 
34 6\\ 
35 6\\ 
36 6\\ 
37 5\\ 
38 5\\ 
39 5\\ 
40 5\\ 
41 4\\ 
42 4\\ 
};

\addplot [Singletonstyle] 
  table[row sep=crcr]{
1 21\\ 
2 20\\ 
3 19\\ 
4 18\\ 
5 17\\ 
6 16\\ 
7 15\\ 
8 14\\ 
9 14\\ 
10 13\\ 
11 13\\ 
12 12\\ 
13 12\\ 
14 11\\ 
15 11\\ 
16 10\\ 
17 10\\ 
18 9\\ 
19 9\\ 
20 8\\ 
21 8\\ 
22 7\\ 
23 7\\ 
24 7\\ 
25 6\\ 
26 6\\ 
27 6\\ 
28 5\\ 
29 5\\ 
30 5\\ 
31 4\\ 
32 4\\ 
33 4\\ 
34 3\\ 
35 3\\ 
36 3\\ 
37 2\\ 
38 2\\ 
39 2\\ 
40 1\\ 
41 1\\ 
42 1\\ 
};

\addplot [LPstyle] 
  table[row sep=crcr]{
1 21\\ 
2 20\\ 
3 19\\ 
4 18\\ 
5 17\\ 
6 16\\ 
7 15\\ 
8 14\\ 
9 14\\ 
10 13\\ 
11 13\\ 
12 12\\ 
13 12\\ 
14 11\\ 
15 10\\ 
16 10\\ 
17 9\\ 
18 9\\ 
19 8\\ 
20 8\\ 
21 8\\ 
22 7\\ 
23 7\\ 
24 6\\ 
25 6\\ 
26 6\\ 
27 5\\ 
28 5\\ 
29 4\\ 
30 4\\ 
31 4\\ 
32 3\\ 
33 3\\ 
34 3\\ 
35 2\\ 
36 2\\ 
37 1\\ 
38 1\\ 
39 1\\ 
40 1\\ 
41 1\\ 
42 1\\ 
};

\addplot [GVstyle] 
  table[row sep=crcr]{
1 21\\ 
2 20\\ 
3 18\\ 
4 17\\ 
5 16\\ 
6 15\\ 
7 14\\ 
8 13\\ 
9 12\\ 
10 11\\ 
11 11\\ 
12 10\\ 
13 9\\ 
14 9\\ 
15 8\\ 
16 8\\ 
17 7\\ 
18 7\\ 
19 6\\ 
20 6\\ 
21 5\\ 
22 5\\ 
23 4\\ 
24 4\\ 
25 4\\ 
26 3\\ 
27 3\\ 
28 3\\ 
29 2\\ 
30 2\\ 
31 2\\ 
32 2\\ 
33 1\\ 
34 1\\ 
35 1\\ 
36 1\\ 
37 1\\ 
38 1\\ 
39 1\\ 
40 1\\ 
41 1\\ 
42 1\\ 
};

\addplot[newCstyle]
  table[row sep=crcr]{
1 21\\ 
2 20\\ 
3 19\\ 
4 18\\ 
5 17\\ 
6 15\\ 
7 14\\ 
9 13\\ 
11 12\\ 
12 11\\ 
13 10\\ 
15 9\\ 
18 8\\ 
21 6\\ 
22 5\\ 
26 4\\ 
31 3\\ 
36 2\\ 
42 1\\ 
}; 

\end{axis}
\end{tikzpicture}
\end{subfigure}\hfill

\centering
\begin{subfigure}[b]{0.5\columnwidth}
\centering
\begin{tikzpicture}[baseline] 
\begin{axis}[%
ymax = 21,
ymin = 0,
xmin = 0,
xmax = 20,
label style={font=\footnotesize},
ymajorgrids,
xmajorgrids,
grid style=dashed,
legend style={at={(1,1)},anchor=north east,font=\scriptsize,legend cell align=left, align=left, draw=black, legend image post style={xscale=0.75,yscale = 0.6,mark size=5pt}, },
ylabel near ticks,
xlabel near ticks,
width = 5.1cm,
legend columns=6,
legend entries={Packing\ , Singleton\ , LP\ , Covering\ , \Cref{const:gcc}},
legend to name={t_GCC_m3_legend},
every axis y label/.style={font=\footnotesize,at={(current axis.north west)},above =1mm,left = 0.5mm},
ylabel={$k$},
ticklabel style = {font=\tiny},
]

\addplot [SPstyle] 
  table[row sep=crcr]{
0 21\\ 
1 18\\ 
2 15\\ 
3 14\\ 
4 12\\ 
5 11\\ 
6 9\\ 
7 8\\ 
8 7\\ 
9 6\\ 
10 5\\ 
11 5\\ 
12 4\\ 
13 3\\ 
14 3\\ 
15 2\\ 
16 2\\ 
17 2\\ 
18 1\\ 
19 1\\ 
20 1\\ 
};

\addplot [Singletonstyle] 
  table[row sep=crcr]{
0 21\\ 
1 19\\ 
2 17\\ 
3 15\\ 
4 14\\ 
5 13\\ 
6 12\\ 
7 11\\ 
8 10\\ 
9 9\\ 
10 8\\ 
11 7\\ 
12 6\\ 
13 6\\ 
14 5\\ 
15 4\\ 
16 4\\ 
17 3\\ 
18 2\\ 
19 2\\ 
20 1\\ 
21 0\\ 
};

\addplot [LPstyle] 
  table[row sep=crcr]{
0 21\\ 
1 18\\ 
2 15\\ 
3 14\\ 
4 12\\ 
5 10\\ 
6 9\\ 
7 7\\ 
8 6\\ 
9 5\\ 
10 4\\ 
11 3\\ 
12 2\\ 
13 2\\ 
14 1\\ 
15 1\\ 
16 1\\ 
17 1\\ 
18 1\\ 
19 1\\ 
20 1\\ 
};

\addplot [GVstyle] 
  table[row sep=crcr]{
0 21\\ 
1 17\\ 
2 13\\ 
3 10\\ 
4 8\\ 
5 6\\ 
6 5\\ 
7 4\\ 
8 3\\ 
9 2\\ 
10 2\\ 
11 1\\ 
12 1\\ 
13 1\\ 
14 1\\ 
15 1\\ 
16 1\\ 
17 1\\ 
18 1\\ 
19 1\\ 
20 1\\ 
};

\addplot[newCstyle]
  table[row sep=crcr]{
0 21\\
1 18\\ 
2 15\\ 
3 13\\ 
4 10\\ 
5 9\\ 
6 7\\ 
7 6\\ 
8 5\\ 
10 4\\ 
11 3\\ 
12 2\\ 
20 1\\ 
}; 

\end{axis}
\end{tikzpicture}
\caption{$q=2$}
\label{fig:m3_q2}
\end{subfigure}
\begin{subfigure}[b]{0.48\columnwidth}
\centering
\begin{tikzpicture}[baseline] 
\begin{axis}[%
ymax = 21,
ymin = 0,
xmin = 0,
xmax = 20,
label style={font=\footnotesize},
xlabel={$t$},
ymajorgrids,
xmajorgrids,
grid style=dashed,
legend style={at={(1,1)},anchor=north east,font=\scriptsize,legend cell align=left, align=left, draw=black},
xlabel near ticks,
width = 5.1cm,
ymajorticks=false,
ticklabel style = {font=\tiny},
every axis x label/.style={font=\footnotesize,at={(current axis.south east)},below =0cm,right = 0mm},
]

\addplot [SPstyle] 
  table[row sep=crcr]{
0 21\\ 
1 19\\ 
2 17\\ 
3 16\\ 
4 15\\ 
5 14\\ 
6 13\\ 
7 12\\ 
8 11\\ 
9 10\\ 
10 10\\ 
11 9\\ 
12 8\\ 
13 8\\ 
14 7\\ 
15 7\\ 
16 6\\ 
17 6\\ 
18 5\\ 
19 5\\ 
20 4\\ 
21 4\\ 
22 3\\ 
23 3\\ 
24 3\\ 
25 2\\ 
26 2\\ 
27 2\\ 
28 1\\ 
29 1\\ 
30 1\\ 
31 1\\ 
};

\addplot [Singletonstyle] 
  table[row sep=crcr]{
0 21\\ 
1 19\\ 
2 17\\ 
3 15\\ 
4 14\\ 
5 13\\ 
6 12\\ 
7 11\\ 
8 10\\ 
9 9\\ 
10 8\\ 
11 7\\ 
12 6\\ 
13 6\\ 
14 5\\ 
15 4\\ 
16 4\\ 
17 3\\ 
18 2\\ 
19 2\\ 
20 1\\ 
21 0\\ 
};

\addplot [LPstyle] 
  table[row sep=crcr]{
0 21\\ 
1 19\\ 
2 17\\ 
3 15\\ 
4 14\\ 
5 13\\ 
6 12\\ 
7 10\\ 
8 9\\ 
9 8\\ 
10 8\\ 
11 7\\ 
12 6\\ 
13 5\\ 
14 4\\ 
15 4\\ 
16 3\\ 
17 2\\ 
18 1\\ 
19 1\\ 
20 1\\ 
};

\addplot [GVstyle] 
  table[row sep=crcr]{
0 21\\ 
1 18\\ 
2 16\\ 
3 14\\ 
4 12\\ 
5 11\\ 
6 9\\ 
7 8\\ 
8 7\\ 
9 6\\ 
10 5\\ 
11 4\\ 
12 4\\ 
13 3\\ 
14 2\\ 
15 2\\ 
16 1\\ 
17 1\\ 
18 1\\ 
19 1\\ 
20 1\\ 
};

\addplot[newCstyle]
  table[row sep=crcr]{
0 21\\
1 19\\ 
2 17\\ 
3 14\\ 
4 13\\ 
5 12\\ 
6 10\\ 
7 9\\ 
8 8\\ 
10 6\\ 
11 5\\ 
12 4\\ 
15 3\\ 
17 2\\ 
20 1\\ 
}; 

\end{axis}
\end{tikzpicture}
\caption{$q=7$}         
\label{fig:m3_q7}
\end{subfigure}\hfill
\vspace{0.2cm}
\pgfplotslegendfromname{t_GCC_m3_legend}
\caption{Dimension $k$ of \Cref{const:gcc} for given minimum distance $d$ (upper row) or given error-correction capability $t$ (lower row). 
Block lengths $\6n=(7,\,7,\,7)$ and scaling factors $\bm{\lambda} = (1,\,2,\,3)$.}
\label{fig:GCC_m3}
\end{figure}
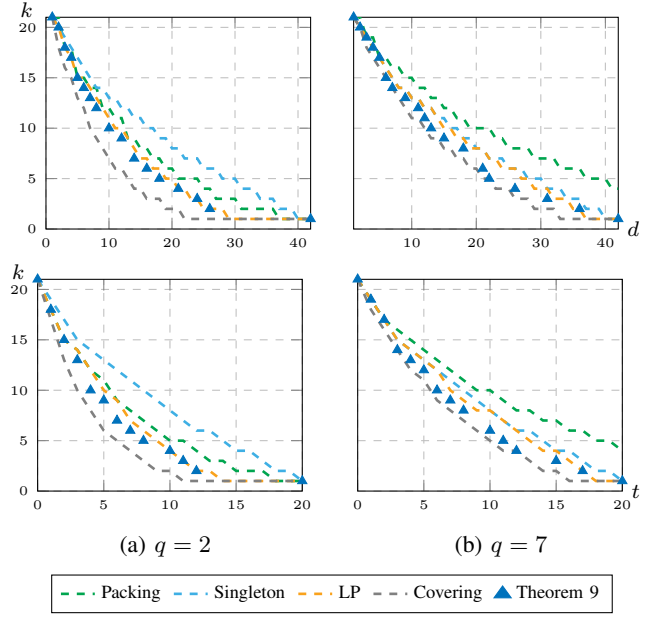

\Cref{fig:GCC_m3} presents parameters achievable by \Cref{const:gcc} for blocks of size $\6n=(7,\,7,\,7)$ and scaling factors $\bm{\lambda} = (1,\,2,\,3)$.
Codes are designed to meet minimum-distance and error-correction requirements, for both $q = 2$ and $q = 7$.
The resulting dimensions attain the covering bound in all cases, and in several cases, also the LP bound is achieved with equality.

\section{Conclusion}

This paper studies the error-correction capability of codes with the weighted-Hamming metric.
By focusing directly on correctable errors rather than the minimum distance, we obtain tighter bounds than the available ones.

Via generalized concatenation, we provide a flexible code construction with guaranteed lower bounds on both minimum distance and error-correction capability.
A modification of the classical decoder efficiently corrects all errors up to this bound, which exceeds half the minimum distance in some instances. 
For short lengths, the resulting codes achieve nearly optimal parameters.

\balance
\bibliographystyle{IEEEtran}
\bibliography{aux/ref}

\end{document}